\pdfoutput=1
\documentclass[sigconf]{acmart}

\definecolor{Green}{rgb}{0.0, 0.5, 0.0}

\definecolor{Purple}{rgb}{128, 0, 128}

\usepackage{soul}
\renewcommand\footnotetextcopyrightpermission[1]{} 
\usepackage[utf8]{inputenc}
\usepackage{comment}
\usepackage{multirow}
\usepackage{pifont}
\usepackage{xifthen}
\usepackage[underline=false]{pgf-umlsd}
\usepackage{tikz}
\usepackage{algorithm}
\usepackage[noend]{algpseudocode}
\usepackage{algpseudocode}
\newcommand{\formatcomment}[1]{\scriptsize\textcolor{blue!25!black}{\texttt{#1}}}

\algrenewcommand{\algorithmiccomment}[1]{\hfill\parbox{5.2cm}{\formatcomment{//\,#1}}}
\algrenewcommand\algorithmicprocedure{\textbf{algorithm}}
\algnewcommand\algorithmicforeach{\textbf{for each}}
\algdef{S}[FOR]{ForEach}[1]{\algorithmicforeach\ #1\ \algorithmicdo}
\usepackage{enumitem}
\usepackage[n ,advantage, operators, sets , adversary, landau, probability, notions, logic, ff , mm, primitives, events, complexity, asymptotics, keys]{cryptocode}
\usepackage{booktabs} 
\usepackage{comment}
\usepackage[normalem]{ulem} 
 \newcommand{\superscript}[1]{\ensuremath{{}^{\textrm{\scriptsize #1}}}}
 
 \newcommand{\mntext}[1]{\colorbox{SkyBlue}{\begin{color}{black}#1\end{color}}}
 \newcommand{\mn}[2][]{{\tiny\superscript{\mntext{\arabic{mn}}}}\marginpar{\scriptsize{
 			\ifthenelse{\isempty{#1}}
 			{\mntext{\parbox{0.95\marginparwidth}{\superscript{\arabic{mn}}~\raggedright{#2}}}}
 			{\mntext{\parbox{0.95\marginparwidth}{\superscript{\arabic{mn}}#1 says~:~\raggedright{#2}}}}
 		}}\stepcounter{mn}}

\usepackage{mathtools}
\usepackage{graphicx}  
\usepackage{xcolor}
\usepackage{tikz}
\usepackage{pgfplots}
\usetikzlibrary{decorations.pathreplacing,calc}

\theoremstyle{definition}
\newtheorem{definition}{Definition}[section]
\newtheorem{theorem}{Theorem}[section]
\newtheorem{corollary}{Corollary}[theorem]
\newtheorem{lemma}[theorem]{Lemma}
\theoremstyle{remark}

\newtheorem{claim}{Claim}[section]

\title{
Equi-Joins over Encrypted Data for Series of Queries
}
\date{}
\usepackage{natbib}

\author{Masoumeh Shafieinejad}
\affiliation{%
  \institution{University of Waterloo}
  \city{Waterloo}
  \country{Canada}}
\email{masoumeh@uwaterloo.ca}

\author{Suraj Gupta}
\affiliation{%
  \institution{University of Waterloo}
  \city{Waterloo}
  \country{Canada}}
\email{sr3gupta@uwaterloo.ca}

\author{Jin Yang Liu}
\affiliation{%
  \institution{University of Waterloo}
  \city{Waterloo}
  \country{Canada}}
\email{jy37liu@uwaterloo.ca}

\author{Koray Karabina}
\affiliation{%
  \institution{National Research Council Canada \& University of Waterloo}
  \city{Waterloo}
  \country{Canada}}
\email{koray.karabina@nrc-cnrc.gc.ca}

\author{Florian Kerschbaum}
\affiliation{%
  \institution{University of Waterloo}
  \city{Waterloo}
  \country{Canada}}
\email{fkerschb@uwaterloo.ca}

\begin{document}
\pagestyle{plain} 
\begin{abstract}
Encryption provides a method to protect data outsourced to a DBMS provider, e.g., in the cloud.
However, performing database operations over encrypted data requires specialized encryption schemes that carefully balance security and performance.
In this paper, we present a new encryption scheme that can efficiently perform equi-joins over encrypted data with better security than the state-of-the-art.
In particular, our encryption scheme reduces the leakage to equality of rows that match a selection criterion and only reveals the transitive closure of the sum of the leakages of each query in a series of queries.
Our encryption scheme is provable secure. We implemented our encryption scheme and evaluated it over a dataset from the TPC-H benchmark.
\end{abstract}
\maketitle

\section{Introduction}\label{Sec:Intro}
Outsourcing data management into the cloud comes with new security risks.
Insiders at the cloud service provider, attackers seeking high-profit targets or international legislators may exploit access to the cloud infrastructure.
Encryption where the key is held at the client provides an additional layer of security countering these threats.
However, regular data management operations such as joins cannot be simply performed over encrypted data.
Hence, specialized encryption schemes for performing joins over encrypted data have been developed \cite{Joins,TCC17,bilinearjoins,CryptDB,pang2014privacy,kerschbaum2013optimal,bloomfilterjoin}.

In this paper we consider only equi-joins, since their security is particularly challenging.
On the one hand, a database management system (DBMS) cannot hide the equality of join attribute values, since the cross product of two tables of size $n$ each, is of size $n^2$ which is prohibitively large for subsequent operations.
Hence, the DBMS needs to select a subset of the cross product using the equality condition.
On the other hand, revealing, the equality condition leaks the frequency of items in a primary key, foreign key join.
This is critical, since primary key, foreign key joins are a very common operation and it has been demonstrated that frequency information is very powerful in cryptanalysis \cite{optimalfrequency,inference}.
The encryption scheme (for joins) by CryptDB \cite{CryptDB} has been effectively broken using this frequency information \cite{inference}.
Consequently, the challenge for any encryption scheme for joins is to allow selecting from the cross-product, yet reveal as few equality conditions as possible.

A state-of-the-art encryption scheme for joins by Hahn et al.~\cite{Joins} further reduces the leakage from deterministic encryption \cite{SQLoED} and onion encryption \cite{CryptDB,AnalQoED} by only leaking equality condition for tuples that match a selection criterion.
However, the leakage of a series of queries in this encryption scheme corresponds to the leakage of the union of the queries, i.e., it may be larger than the union (sum) of the leakage of each query.
We provide an example of such super-additive leakage in Section~\ref{Sec:Problem}.
In this paper we aim to reduce the leakage of equality conditions even further.

We present a new encryption scheme for joins that not only restricts the leakage of the equality condition to tuples that match a selection criterion, but also limits the leakage of a series of queries corresponds to the transitive closure of the union of the leakage of each query, i.e., there is no super-additive leakage.
We believe that this leakage is a natural lower bound for the leakage of an efficient encryption scheme for non-interactive joins using one outsourced DBMS.
Note that oblivious joins \cite{sovereignjoins,obliviousquery,simeon} which only leak the size of the joint table, either require secure hardware or multi-party computation \cite{smcql}, and an interactive protocol that reveals the size of the joint table.

Our construction requires the use of a new cryptographic technique -- function-hiding inner product encryption -- compared to previous approaches.
Our construction is efficient with cryptographic operations requiring only a few milliseconds and the ability to run hash-based joins with expected time complexity $O(n)$.
Our construction works for arbitrary equi-joins.
As a comparison, the state-of-the-art encryption scheme by Hahn et al.~\cite{Joins} requires nested-loop joins (with $O(n^2)$ time complexity), and it only works for primary key, foreign key joins.
We implemented our encryption scheme and evaluated encrypted joins over a dataset from the TPC-H benchmark. 

In summary, our contributions are as follows:
\begin{itemize}
    \item We provide a new encryption schemes for non-interactive equi-joins over encrypted data where a series of queries only leaks the transitive closure of the union of the leakage of each query, i.e., without super-additive leakage.
    \item We analyze the security of our scheme using a formal security proof.
    \item We evaluate the performance of a DBMS using our encryption scheme over a database from the TPC-H benchmark.
\end{itemize}

The remainder of the paper is structured as follows:
Section~\ref{Sec:Problem} describes the system model and problem in detail.
Section~\ref{Sec:Prelim} provides the necessary cryptographic background and Section~\ref{Sec:Protocol} describes our join encryption scheme.
We presents its security proof in Section~\ref{Sec:Security} and its performance evaluation in Section~\ref{Sec:Experiments}.
Section~\ref{Sec:RelatedWork} surveys related work and Section~\ref{Sec:Conclusions} summarizes our conclusions.
\section{System Model}\label{Sec:Problem}
In outsourced data management, a client stores their [sensitive] data on a database server under the control of a DBMS service provider~\citep{hacigumus2002providing}. Later, the client can access the outsourced data through an online query interface provided by the server. Clients desire to allow the server to process data queries while maintaining the confidentiality of the data. For this purpose, they encrypt data before outsourcing. However, encrypted data is hard to process. Therefore, to allow for  more  expressive server-side data processing, the client  will provide certain ``unlocking'' information (tokens) for a set of specific (equi-join) predicates. The clients expects the server to behave semi-honestly and perform exactly the considered query while trying to find out any additional information. We consider a relational model, where the client outsources their data in a number of (at least two) tables each consisting of several data columns (e.g., relational attributes). In this work we focus on performing equi-join over two outsourced tables $T_A$ and $T_B$. Without loss of generality, we assume both tables have $n$ rows and $m$ attributes, with each attribute taking its values from a domain of size $\ell$, to simplify the notations. The equi-join result on the two join columns of the table-pair $(T_A, T_B)$, is a subset of the cross-product of rows from the two tables that contain equal values in their join columns \citep{Joins}. Assume table $T_A$ with $|T_A|= n$ records, has schema $(A_0, A_1, \cdots, A_m)$ with join key $A_0$ that identifies the join column and other attributes $A_1, \cdots, A_m$. Each attribute $A_i$ has a domain $X_i$. We denote by $x_{i,j}$, $i \in [m]$ and $j \in [\ell]$ , the $j^{th}$ domain value in $X_i$. We show the rows in $T_A$ by tuples of variables: $(a_0^1, \cdots, a_m^1), \cdots, (a_0^n, \cdots, a_m^n)$. Similarly, $T_B$ has schema $(B_0, B_1, \cdots, B_n)$, with domain $Y_i$ for each $B_i$. We denote by $y_{i,j}$,$i \in [m]$ and $j \in [\ell]$, the $j^{th}$ value in $Y_i$. $T_B$ has $|T_B|=n$ rows shown by $(b_0^1, \cdots, b_m^1), \cdots, (b_0^n, \cdots, b_m^n)$. The equi-join with join attributes $A_0$ and $B_0$ is an operation on tables $T_A$ and $T_B$, denoted by\footnote{As used in \citep{Joins}} $T_A \bowtie T_B$. The result of $T_A \bowtie T_B$ has schema $(\Theta, A_1, \cdots, A_m, B_1, \cdots, B_m)$, and consists of records $(\theta^{r,r^{\prime}}, a_1^r, \cdots, a_m^r, b_1^{r^{\prime}}, \cdots, b_m^{r^{\prime}})$. The attribute $\Theta$ takes its values from all $A_0$'s and $B_0$'s that match, in other words: $\theta^{r,r^{\prime}} = a_0^r = b_o^{r^{\prime}}$, for all $r \in [n]$ , $r^{\prime} \in [n]$ where $a_0^r=b_0^{r^{\prime}}$ holds. There is a further filtering based on additional filtering-predicates chosen from  $\{A_1, \cdots, A_m\}$ and $\{B_1, \cdots, B_m\}$.

\subsection{Problem Description} 
We describe the secure join problem through the following example. 
\begin{example}\label{ex:problem}
Consider $T_A$ and $T_B$ in Tables \ref{Tab:Teams} and \ref{Tab:Employee} containing employees information and their teams, respectively. Thus: $(A_0, A_1)$ = (Key, Name), and  $(B_0, B_1, B_2, B_3)$ = (Team, Record, Employee, Role). Assume the filtering-predicates are chosen over $A_1$ with domain $X_1=$ \{Web application, Database\} for $T_A$, and $B_3$ with domain $Y_3=$ \{Programmer, Tester\} for $T_B$. 
\begin{table}[h]
  \begin{tabular}{| c | c |}
    \hline
     Key & Name \\ \hline
     1 & Web Application  \\ \hline
     2 & Database \\ \hline
  \end{tabular}
  \caption{Teams}
  \label{Tab:Teams}
\end{table}
\begin{table}[h]
  \begin{tabular}{| c | c | c | c |}
    \hline
     Record & Employee & Role & Team \\ \hline
     1 & Hans & Programmer & 1 \\ \hline
     2 & Kaily & Tester & 1 \\ \hline
     3 & John & Programmer &  2\\ \hline
     4 & Sally & Tester &  2 \\ \hline
  \end{tabular}
  \caption{Employees}
  \label{Tab:Employee}
\end{table}
The equi-join of these two tables over the join keys $A_0 = $ Key and $B_0 = $ Team, includes four pairs $(a_0^r, b_0^{r^{\prime}})$ from row $r$ in $T_A$ and row $r^{\prime}$ in $T_B$, with true equality condition: $(a_0^1, b_0^1)$, $(a_0^1, b_0^2)$, $(a_0^2, b_0^3)$, $(a_0^2, b_0^4)$. Two additional equality pairs $(b_0^1, b_0^2)$, $(b_0^3, b_0^4)$ only from Table $T_B$ need to be included in order to complete the transitive closure.

We also emphasize that the uniqueness of the join attribute values in $A_0$ is a feature of Example \ref{ex:problem}, not a general requirement, since our scheme is not limited to joins between primary key,  foreign key joins. 

We consider three database operations (queries) at times $t_0 < t_1 < t_2$, i.e. $t_0$ is the point in time after encrypted database upload, $t_1$ is the point in time after the first query, but before the second query and $t_2$ is the point in time after the first and the second query.
\begin{itemize}
    \item[$t_0$:] \sf{Encrypted database upload.}
    \item[$t_1$:] \sf{SELECT $\star$ FROM Employees JOIN Teams ON Team = Key
    \emph{WHERE} Name = ``Web Application'' AND Role = ``Tester''}
    \item[$t_2$:] \sf{SELECT $\star$ FROM Employees JOIN Teams ON Team = Key
    \emph{WHERE} Name = ``Database'' AND Role = ``Programmer''}
\end{itemize}

\begin{table}[h]
  \begin{tabular}{| c | c | c | c | c |}
    \hline 
    Record& Employee & Role & T.Key & T.Name\\ \hline
    2 & Kaily & Tester & 1 & Web Application\\ \hline
  \end{tabular}
  \caption{The result of equi-join query at $t_1$}
  \label{Tab:Join1}
\end{table}

\begin{table}[h]
  \begin{tabular}{|c| c | c | c | c |}
    \hline 
    Record & Employee & Role & T.Key & T.Name\\ \hline
    3 & John & Programmer & 2 & Database\\ \hline
  \end{tabular}
  \caption{The result of equi-join query at $t_2$}
  \label{Tab:Join2}
\end{table}
In our analysis we compare encryption schemes for joins over encrypted data based on the (number of) pairs with true equality condition they reveal.
The results of the two queries are depicted in Tables \ref{Tab:Join1} and \ref{Tab:Join2}, respectively.
To compute those results efficiently the DBMS needs to reveal the equality condition of the pairs $(a_0^1, b_0^2)$ and $(a_0^2, b_0^3)$, i.e., this represents our minimum leakage and no efficient encryption scheme using non-interactive matches on a single DBMS can avoid this leakage.
\end{example}

The first proposal for database operations over encrypted data by Hacig\"um\"us et al.~\cite{SQLoED} used deterministic encryption \cite{deterministic1,deterministic2} to compute joins.
In deterministic encryption each data value is deterministically encrypted to the same ciphertext, such that the DBMS can compare the ciphertexts for an equi-join.
While this idea by itself has been shown to be insecure \cite{inference}, it is still the fundamental idea for subsequent schemes.
In our analysis, we can state that deterministic encryption reveals all six (equal) pairs at time $t_0$.

An improvement over deterministic encryption was presented by CryptDB \cite{CryptDB}.
CryptDB uses onion encryption and wraps each deterministic ciphertext in a probabilistic ciphertext.
Hence, at time $t_0$ no pair is revealed, but at time $t_1$ all six pairs are revealed, since the wrapped probabilistic encryption needs to be stripped before an equi-join is feasible.
The equality condition can be restricted to a few columns by using re-encryptable deterministic encryption \cite{kerschbaum2013optimal,TCC17}.
However, in our example there are only two columns and both are involved in the same join operation.
Specialized schemes, such as \cite{bilinearjoins,bloomfilterjoin,pang2014privacy}, also maintain a re-encryption token for the entire table covered by the pair of columns.
Hence, they do not offer any improvement against our analysis of CryptDB.

A state-of-the-art encryption scheme for joins over encrypted data by Hahn et al.~\cite{Joins}, loosely speaking, replaces the probabilistic encryption by key-policy attribute-based encryption (KP-ABE)~\cite{kpabe}.
They also use searchable encryption instead of deterministic encryption, but we ignore this in our paper, since it does not impact our analysis.
KP-ABE ensures that only rows that match a selection criterion specified by the attributes, can be decrypted and hence their ciphertext can be matched.
After, time $t_1$ this scheme reduces the revealed pairs to $(a_0^1, b_0^2)$ which is the minimum at this point of time.

However, consider what happens in this encryption scheme at time $t_2$.
In the first query, the wrapped KP-ABE encryption on rows 1 from Team and 2 from Employees, but also row 4 from Employees since it also matches Role = ``Tester'', are removed.
In the second query, the wrapped KP-ABE encryption on rows 2 from Teams and 3 Employees, but also row 2 from Employees since it also matches Role = ``Programmer'', are removed.
In summary, at time $t_2$ the wrapped probabilistic encryption on all rows has been removed and all six (equal) pairs are revealed, since the adversary controlling the DBMS can match the unwrapped rows.
This reveals more pairs than necessary for the union of the queries.
We call this super-additive leakage, since it is more than the sum of the leakages of each query. 

The challenge for our encryption scheme is to only reveal the pairs $(a_0^1, b_0^2)$ and $(a_0^2, b_0^3)$ at time $t_2$.
The goal is a leakage equal to the transitive closure over the union of the leakages of each query.
Hence, we claim that an encryption scheme with this leakage is more secure under a sequence of queries than the state-of-the-art encryption scheme by Hahn et al.
Informally speaking, we aim for re-encrypting the deterministic ciphertexts to different keys for each query when the probabilistic encryption is removed.
Constructing such an encryption scheme is not trivial, but we claim that a modification to function-hiding inner-product encryption \cite{InnerProductEnc} in combination with an encoding scheme using polynomials achieves the desired property.
\section{Preliminaries}\label{Sec:Prelim}
\subsection{Bilinear Groups}\label{Sec:Bilinear}
Let $\mathbb{G}_1$ and $\mathbb{G}_2$ be two distinct groups of prime order $q$, and let $g_1 \in \mathbb{G}_1$ and $g_2 \in \mathbb{G}_2$ be generators of the respective groups. 
Let $e:\mathbb{G}_1 \times \mathbb{G}_2 \rightarrow \mathbb{G}_T$ be a function that maps two elements from $\mathbb{G}_1$ and $\mathbb{G}_2$ onto a target group $\mathbb{G}_T$, which is also of prime order $q$. We follow the style of Kim et al.~\citep{InnerProductEnc} to write the
group operation in $\mathbb{G}_1$, $\mathbb{G}_2$ and $\mathbb{G}_T$ multiplicatively and write 1 to denote their multiplicative identity. The tuple $(\mathbb{G}_1, \mathbb{G}_2, \mathbb{G}_T, q, e)$ is an asymmetric bilinear group \citep{BilinearGroups1,BilinearGroups2,BilinearGroups3} if these properties hold:
\begin{itemize}
    \item The group operation in the groups $\mathbb{G}_1$, $\mathbb{G}_2$, $\mathbb{G}_T$ and the mapping $e$, are all efficiently computable.
    \item The map $e$ is non-degenerate: $e(g_1, g_2) \neq 1$.
    \item The map $e$ is bilinear: for all $x,y \in \mathbb{Z}_q$, this holds:\\ $e(g^x_1, g^y_2)=e(g_1,g_2)^{xy}$
\end{itemize}
When dealing with vectors of group elements, for a group $\mathbb{G}$ of prime order $q$ with an element $g \in \mathbb{G}$ and a row vector $\mathbf{v}= (v_1, \cdots, v_n) \in \mathbb{Z}_q^n$ where $n \in \mathbb{N}$, we write $g^{\mathbf{v}}$ to denote the vectors of group elements $(g^{v_1}, \cdots, g^{v_n})$. Also, for any scalar $k \in \mathbb{Z}_q$ and vectors $\mathbf{v,w} \in \mathbb{Z}_q^n$, we write: $(g^\mathbf{v})^k=g^{(\mathbf{v}k)}$ and $g^{\mathbf{v}}.g^{\mathbf{w}}=g^{\mathbf{v}+\mathbf{w}}$. Furthermore, the pairing operation over groups is shown as: $e(g_1^{\mathbf{v}},g_2^{\mathbf{w}}) = \Pi_{i \in [n]}e(g_1^{vi}, g_2^{w_i}) =  e(g_1,g_2)^{\langle \mathbf{v,w}
\rangle}$.
\subsection{Polynomial Functions}\label{Sec:PolyFunc}
We use the definition of polynomial functions by Leung et al.~\citep{Polynomials}. Consider a polynomial $f(x) = a_nx^n+a_{n-1}x^{n-1} +\cdots+a_1x+a_0$ of the set of polynomial in $x$ over the prime field $\mathbb{Z}_q$. If the indeterminate $x$ in the expression is regarded as a \emph{variable} which can assume any value in $\mathbb{Z}_q$, then in the most natural way the polynomial $f(x)$ will give rise to a mapping of the set $\mathbb{Z}_q$ into $\mathbb{Z}_q$. This mapping is defined by the polynomial $f(x)$ as follows. To each element $c$ of the domain $\mathbb{Z}_q$ there corresponds under the mapping the unique value $f(c) = a_nc^n+a_{n-1}c^{n-1} +\cdots+a_1c+a_0$ of the range $\mathbb{Z}_q$. Thus this is the mapping $c \rightarrow f(c)$ of $\mathbb{Z}_q$ into $\mathbb{Z}_q$. This mapping is called the \emph{polynomial function} in the variable $x$ defined by the polynomial $f(x)$. 
Polynomial functions that are bounded by Lemma \ref{lem:Schwartz-Zippel} in their probability of evaluating to zero. 
\begin{lemma}{(Schwartz-Zippel \citep{Schwartz,Zippel} adapted in \citep{InnerProductEnc})} \label{lem:Schwartz-Zippel}
Fix a prime $q$ and let $f \in \mathbb{Z}_q[x_1, \cdots, x_n]$ be an n-variate polynomial with total degree at most $t$ and which is not identically zero. Then,
\begin{equation}
    Pr[x_1, \cdots, x_n \xleftarrow{R}\mathbb{Z}_q:f(x_1, \cdots, x_n)=0] \leq \frac{t}{q}.
\end{equation}
\end{lemma}

\subsection{Function-Hiding Inner Product Encryption}\label{Sec:FHIPE}
We build our scheme on the function-hiding inner-product encryption construction by Kim et al.~\cite{InnerProductEnc}. Their scheme consists of four algorithms $\Pi_{ipe}$  = (IPE.Setup, IPE.KeyGen, IPE.Encrypt, IPE.Decrypt) described below. We keep their notations where bold lowercase letters (e.g. \textbf{v,w}) denotes vectors and bold uppercase letters (e.g. $\mathbf{B , B^*}$) denote matrices. $\mathbb{GL}_n(\mathbb{Z}_q)$ is the general linear group of ($n \times n$) matrices over $\mathbb{Z}_q$.
\begin{enumerate}
    \item IPE.Setup$(1^\lambda, S)$:  On input of the security parameter $\lambda$ and $S$ a polynomial-sized (in $\lambda$) subset of $\mathbb{Z}_q$, the setup algorithm samples an asymmetric bilinear group $(\mathbb{G}_1,\mathbb{G}_2,\mathbb{G}_T,q,e)$ and chooses generators $g_1 \in \mathbb{G}_1$ and $g_2 \in \mathbb{G}_2$.  Then, it samples $\mathbf{B} \gets \mathbb{GL}_n(\mathbb{Z}_q)$ and sets $\mathbf{B^*} = det(\mathbf{B}) \cdot (\mathbf{B}^{-1})^T$. Finally, the setup algorithm outputs the public parameters $pp= (\mathbb{G}_1,\mathbb{G}_2,\mathbb{G}_T,q,e)$ and the master secret key $msk= (pp,g_1,g_2,\mathbf{B},\mathbf{B^*})$. 
    \item IPE.KeyGen(msk, \textbf{v}): On input of the  master  secret  key msk and  a  vector $\mathbf{v} \in \mathbb{Z}^n_q$,  the  key generation algorithm chooses a uniformly random element $\alpha \xleftarrow{R} \mathbb{Z}_q$ and outputs the pair: $sk = (K_1, K_2) = (g_1^{\alpha \cdot det(\mathbf{B})}, g_1^{\alpha \cdot \mathbf{v} \cdot \mathbf{B}})$. Note that the second component is a vector of group elements.
    \item IPE.Encrypt(msk, \textbf{w}): On input of the  master  secret  key msk and  a  vector $\mathbf{w} \in \mathbb{Z}^n_q$,  the encryption algorithm chooses a uniformly random element $\beta \xleftarrow{R} \mathbb{Z}_q$ and outputs the pair: $C = (C_1, C_2) = (g_2^{\beta}, g_2^{\beta \cdot w \cdot \mathbf{B^*}})$.
    \item IPE.Decrypt(pp, sk, ct):  On input of the public parameters pp, a secret key $sk= (K_1,K_2)$ and a ciphertext $C= (C_1,C_2)$, the decryption algorithm computes $D_1=e(K_1,C_1)$ and $D_2=e(K_2,C_2)$. Then, it checks whether there exists $z \in S$ such that $(D_1)^z = D_2$. If so, the decryption algorithm outputs $z$. Otherwise, it outputs $\perp$. The efficiency of this algorithm is guaranteed by $|S| =poly(\lambda)$.
\end{enumerate}
The correctness of $\Pi_{ipe}$ holds when the plaintext vectors \textbf{v} and \textbf{w} satisfy $\langle \mathbf{v, w} \rangle \in S$ for a polynomially-sized $S$. Since $D_1 = e(K_1, C_1) = e(g_1,g_2)^{\alpha \beta \cdot det(\mathbf{B})}$ and $D_2 = e(K_2, C_2) = e(g_1, g_2)^{\alpha \beta \cdot \mathbf{v} \cdot \mathbf{B}(\mathbf{B^*})^T\mathbf{w}^T} = e(g_1,g_2)^{\alpha \beta  \cdot det(\mathbf{B}) \cdot \langle \mathbf{v, w} \rangle}$, the decryption algorithm will correctly output $\langle \mathbf{v, w} \rangle$ if $\langle \mathbf{v, w} \rangle \in S$.
\renewcommand{\mess}[4][0]{
  \stepcounter{seqlevel}
  \path
  (#2)+(0,-\theseqlevel*\unitfactor-0.7*\unitfactor) node (mess from) {};
  \addtocounter{seqlevel}{#1}
  \path
  (#4)+(0,-\theseqlevel*\unitfactor-0.7*\unitfactor) node (mess to) {};
  \draw[,] (mess from) (mess to) node[midway, above]
  {#3};
  \node (\detokenize{#3} from) at (mess from) {};
  \node (\detokenize{#3} to) at (mess to) {};
}
\newcommand{\sdinit}{%
   \pgfdeclarelayer{umlsd@background}%
   \pgfdeclarelayer{umlsd@threadlayer}%
   \pgfsetlayers{umlsd@background,umlsd@threadlayer,main}%
}
\newcommand{\sdbegin}{%
   \setlength{\unitlength}{1cm}%
   \tikzstyle{sequence}=[coordinate]%
   \tikzstyle{inststyle}=[rectangle, draw, anchor=west, minimum
   height=0.8cm, minimum width=1.6cm, fill=white, 
   drop shadow={opacity=1,fill=black}]%
   \ifpgfumlsdroundedcorners%
      \tikzstyle{inststyle}+=[rounded corners=3mm]%
   \fi%
   \tikzstyle{blockstyle}=[anchor=north west]%
   \tikzstyle{blockcommentstyle}=[anchor=north west, font=\small]%
   \tikzstyle{dot}=[inner sep=0pt,fill=black,circle,minimum size=0.2pt]%
   \global\def\unitfactor{0.6}%
   \global\def\threadbias{center}%
   \setcounter{preinst}{0}%
   \setcounter{instnum}{0}%
   \setcounter{threadnum}{0}%
   \setcounter{seqlevel}{0}%
   \setcounter{callevel}{0}%
   \setcounter{callselflevel}{0}%
   \setcounter{blocklevel}{0}%
   \node[coordinate] (inst0) {};%
}
\newcommand{\sdend}{%
   \begin{pgfonlayer}{umlsd@background}%
      \ifnum\value{instnum}>0%
         \foreach \t [evaluate=\t] in {1,...,\theinstnum}{%
            \draw[dotted] (inst\t) -- ++(0,-\theseqlevel*\unitfactor-2.2*\unitfactor);%
         }%
         
      \fi%
      \ifnum\value{threadnum}>0%
         \foreach \t [evaluate=\t] in {1,...,\thethreadnum}{%
            \path (thread\t)+(0,-\theseqlevel*\unitfactor-0.1*\unitfactor) node (threadend) {};%
            \tikzstyle{threadstyle}+=[threadcolor\t]%
            \drawthread{thread\t}{threadend}%
         }%
      \fi%
   \end{pgfonlayer}%
}

\section{Protocol Overview}\label{Sec:Protocol}
Let tables $T_A$ and $T_B$ be the tables to be encrypted, and joined over a set of rows selected based on their attribute values. We propose a protocol that achieves this goal securely. In this section, we first describe our implementation of the selection operation, then we explain the modifications we made to the function-hiding inner product encryption scheme of Section \ref{Sec:FHIPE}, in order to implement a combined \textit{selection and joins} operation. Subsequently in Section \ref{Sec:Scheme}, we provide a full picture of our scheme in details. We clarify the steps of our protocol described in this section with the aid of Example \ref{ex:Sample}. 

\begin{example}\label{ex:Sample}
Assume sample rows $r$ and $r^{\prime}$ of the tables $T_A$ and $T_B$ indicated in Tables \ref{Tab:Table_A} and \ref{Tab:Table_B}. 
\begin{table}[h]
  \centering
  \begin{tabular}{|>{\centering}p{0.07\textwidth}|>{\centering}p{0.07\textwidth}|>{\centering\arraybackslash}p{0.07\textwidth}|}
    \hline & &\\[-1em]
     $A_0$   & $A_1$   &  $A_2$      \\ [1ex]\hline  & &\\[-1em]
    $a_0^r$       & $a_1^r$        & $a_2^r$  \\ [1ex]\hline
  \end{tabular}
  \caption{A sample row $r$ in $T_A$}
  \label{Tab:Table_A}
  \hfill
\end{table}
\begin{table}[h]
  \centering
  \begin{tabular}{|>{\centering}p{0.07\textwidth}|>{\centering}p{0.07\textwidth}|>{\centering\arraybackslash}p{0.07\textwidth}|}
    \hline & &\\[-1em]
     $B_0$   & $B_1$   &  $B_2$      \\ [1ex]\hline  & &\\[-1em]
    $b_0^{r^{\prime}}$ & $b_1^{r^{\prime}}$  & $b_2^{r^{\prime}}$  \\ [1ex]\hline
  \end{tabular}
  \caption{A sample row $r^{\prime}$ in $T_B$}
  \label{Tab:Table_B}
\end{table}

Consider the following database equi-join query, with specified selection filters:
\begin{center}
\sf{SELECT $\star$ FROM $T_A$ JOIN $T_B$ ON $A_0 = B_0$ WHERE\\ $A_1$ IN $\Phi_1 = (\phi_{1,1}, \ldots, \phi_{1,t})$ AND $B_1$ IN $\Psi_1 = (\psi_{1,1}, \ldots, \psi_{1,t})$.}
\end{center}

This query results in the join of the sample rows in Tables \ref{Tab:Table_A} and \ref{Tab:Table_B}, if $a_0^r==b_0^{r^{\prime}}$, and if the specified selection criterion matches the values of $a_1^r$ and $b_1^{r^{\prime}}$; in other words: 
\begin{center}
$\exists \phi_{1,z} \in \Phi_1$  s.t. $a_1^r=\phi_{1,z} \land \exists \psi_{1,z} \in \Psi_1$ s.t. $b_1^{r^{\prime}}=\psi_{1,z}$.    
\end{center}
\end{example}
The join query's \emph{WHERE} clause in Example \ref{ex:Sample}, imposes $t$ restrictions ($\Phi_1$) on one attribute ($A_1$) in $T_A$, and $t$ restrictions ($\Psi_1$) on one attribute ($B_1$) in $T_B$. As we describe in Section \ref{sec:polynomials}, in general, the IN clause for table $T_{\tau}$, $\tau \in \{A,B\}$, can impose a maximum of $t$ restrictions on each of the $m$ attributes in table $T_{\tau}$.   

\subsection{Encoding Selection Operations in Polynomials}\label{sec:polynomials}
We implement the selection operation through the usage of polynomial functions. Each polynomial $P_i$, $i \in [m]$, is of degree $t$ and can encode maximum $t$ attribute values, as its roots. These attribute values are specified by $\Phi_i$ in the \emph{IN} clause of the join query for $T_A$, and all belong to the same domain $X_i$. Recall from Section \ref{Sec:Problem} that $X_i$ is the domain of the attribute $A_i$. Hence, $P_i$ enables the query to select the rows from $T_A$ that have \textit{particular} attribute values as attribute $A_i$. Similarly, each polynomial $Q_i$, $i \in [m]$, takes the values specified in the \emph{IN} clause $\Psi_i$ for $T_B$, as its roots. These values belong to $Y_i$, which is the domain of the attribute $B_i$. Therefore, $Q_i$ enables the join query to select the rows from $T_B$ with \textit{desired} attribute values for $B_i$. Both $P_i(x)$ and $Q_i(y)$ take their coefficients from $\mathbb{Z}_q$.

Thus: $P_i(x) = \sum_{j=0}^{t} p_{i, j} \cdot x^j$, such that $P_{i}(\phi_{i,z})=0$, $z \in [t]$. In a similar way, $Q_i(y)=\sum_{j=0}^{t} q_{i, j} \cdot y^j$, $i \in [m]$, such that $Q_{i}(\psi_{i,z})=0$, $z \in [t]$. We emphasize that with the requirements of degree $t$, and maximum $t$ specified points (roots), each $P_i$ or $Q_i$ can take any polynomial from a set of at least $q$ distinct polynomials. If an attribute $A_i$ or $B_i$ is not included the selection criterion, it is encoded as the zero polynomial; i.e.~$P_i = 0$ or $Q_i = 0$ respectively. 
Otherwise, according to Lemma \ref{lem:Schwartz-Zippel}, the probability of the non-zero polynomials of form $P_i$ or $Q_i$ evaluating to zero at a randomly selected point is bounded by $\frac{t}{q}$. We assume an efficient and injective embedding from the attribute values of $A_i$'s and $B_i$'s, $i \in [m]$ , to $\mathbb{Z}_q$ which generates elements in $\mathbb{Z}_q$ uniformly at random, to comply with the Schwartz-Zippel lemma. We use a cryptographic hash function to provide such a mapping.

To apply the filtering predicates in $\Phi_i$'s and $\Psi_i$'s specified in the query, the client sends the corresponding polynomials' coefficients $p_{i,j}$'s and $q_{i,j}$'s as join tokens to the server. The server multiplies these coefficients by their corresponding [pre-stored] powers of attribute values. When these polynomials evaluating to zeros at rows with the target attribute values, they unwrap the join value for the server. 

\begin{example}\label{ex:PolyEncoding}
Consider $T_A$ and $T_B$ in Tables \ref{Tab:Table_A} and \ref{Tab:Table_B}, in addition to the join query in Example \ref{ex:Sample}. The client uploads the non-join attribute values $a_1^r$ and $a_2^r$ of $T_A$ on the server, as [an encrypted] vector $((a_1^r)^0, \cdots, (a_1^r)^t$, $(a_2^r)^0, \cdots, (a_2^r)^t)$. Similarly the non-join attribute values of $T_B$, i.e. $b_1^{r^{\prime}}$ and $b_2^{r^{\prime}}$, are stored as: [encrypted] $((b_1^{r^{\prime}})^0, \cdots, (b_1^{r^{\prime}})^t$, $(b_2^{r^{\prime}})^0, \cdots, (b_2^{r^{\prime}})^t)$. At query time, the client selects their filtering predicates, such as $\Phi_1$ for attribute $A_1$, and $\Psi_1$ for attribute $B_1$. Hence, the client chooses $P_1(x)$ and $Q_1(y)$ such that they evaluate to zero at $\phi_{1,z}$'s, $z \in [t]$ and $\psi_{1,z}$'s, $z \in [t]$ respectively. For the other attributes, i.e.~ $a_2^r$ and $b_2^{r^{\prime}}$, the client assigns polynomials that are identical to zero. Thus, the client's join query token, consists [encrypted] $( p_{10}, \cdots, p_{1t}, \Vec{0})$ for $T_A$ and [encrypted] $(q_{10}, \cdots, q_{1t}, \Vec{0})$ for $T_B$, where $\Vec{0}$ is an all-zero vector of length $t+1$. It is easy to see that the inner product of the token generated for $T_A$ and the stored vector for $T_A$ equals zero, if $a_1^r=\phi_{1,z}$, for a $\phi_{1,z} \in \Phi_1$. A similar argument holds for $T_B$ and $b_1^{r^{\prime}}=\psi_{1,z}$, for a $\psi_{1,z} \in \Psi_1$. To give a clear picture of the polynomial-encoding implementation in this example, we skipped the details of the encryption operation. We provide a full description of our scheme in Section \ref{Sec:Scheme}. 
\end{example}
\subsection{Modified Function-hiding IPE}\label{Sec:Modifications}
We described the function-hiding inner-product encryption construction $\Pi_{ipe}$ by Kim et al.~[2] in Section \ref{Sec:FHIPE}, which is the base for our scheme. However, we made the following modifications on the construction to adjust it with our scheme's needs.
\begin{enumerate}
    \item We set the random parameters in IPE.KeyGen and IPE.Encrypt to 1, i.e.~$\alpha = \beta =1$. We instead incorporate random parameters $\delta$ and $\gamma$ in the input vectors \textbf{v} and \textbf{w} of IPE.Key and IPE.Encrypt, making vectors are of forms $\mathbf{v} = (\mathbf{v^{\prime}}, 0, \delta)$ and $\mathbf{w} = (\mathbf{w^{\prime}}, \gamma, 0)$.
    \item $\Pi_{ipe}$ generates a pair of secret keys $sk = (K_1, K_2)$ during IPE.KeyGen, a pair of encrypted outputs $C=(C_1, C_2)$ in IPE.Encrypt that decrypt to the pair $(D_1, D_2)$ in IPE.Decrypt. In our scheme, we just use one element of these pairs. Hence, as we describe in full details in Section \ref{Sec:Scheme}, we use the following parameters in our scheme: $sk\footnote{We show this value by Tk in our scheme, as it acts as an unlocking token.}  = g_1^{\mathbf{v} \cdot B}$, $C = g_2^{\mathbf{w} \cdot B^{\star}}$, $D = e(g_1, g_2)^{det(B) \cdot \langle \mathbf{v} , \mathbf{w} \rangle}$.
    \item $\Pi_{ipe}$ extracts the value of $\langle \mathbf{v}, \mathbf{w} \rangle$ from the decrypted value $D$ at the end of the protocol. We are not interested in obtaining the value of $\langle \mathbf{v}, \mathbf{w} \rangle$ in our scheme, but a deterministic function of this value. Hence, we do not require the  $\langle \mathbf{v}, \mathbf{w} \rangle$ reside in the polynomial-sized subset $S$ for which one can break the discrete logarithm with overwhelming probability. Ultimately, We apply $\Pi_{ipe}$ twice independently in our join encryption scheme, first on table $T_A$ and then on table $T_B$\footnote{The order does not matter here.}. We calculate $D(sk,C)$ for both of these iteration and conclude a ``match'' if the obtained D's are equal.
\end{enumerate}

\subsection{Our Secure Join Scheme}\label{Sec:Scheme}
\begin{figure*}[h!]
   \centering
   \sdinit{}
   \begin{tikzpicture}
      \sdbegin{}
      \newinst{A}{$T_A$}
      \newinst[5]{B}{$T_B$} 
      \begin{sdblock}{Upload Phase}{}
      \begin{sdblock}{Setup (Client)}{}
      \mess{B}{\shortstack[c]{ }}{A}
      \node[anchor=center] at (mess from) {\shortstack[l]{
      $pp= (\mathbb{G}_1,\mathbb{G}_2,\mathbb{G}_T,q,e)$\\
      $msk= (pp,g_1,g_2,\mathbf{B},\mathbf{B^*})$\hspace{9.8em}
      }};
      \node[anchor=center] at (mess to) {\shortstack[l]{
      $pp= (\mathbb{G}_1,\mathbb{G}_2,\mathbb{G}_T,q,e)$\\
      $msk= (pp,g_1,g_2,\mathbf{B},\mathbf{B^*})$\hspace{9.8em}
      }};      
      \end{sdblock}
      \begin{sdblock}{Encryption (Client)}{}
      \mess{A}{\shortstack[c]{ }}{B}
      \node[anchor=center] at (mess from) {\shortstack[l]{\vspace{3.5em} \\
      $\gamma_{A,1}^r, \gamma_{A,2}^r \xleftarrow{R} \mathbb{Z}_q$\\
      $C_{A}^r = g_2^{\mathbf{w}_{A}^r \mathbf{B^{\star}}}$\\
      $\mathbf{w}_{A}^r = (\boldsymbol{\omega}^r_A, \gamma_{A,1}^r, 0 )$\\
      $\boldsymbol{\omega}^r_{A} = (H(a_0^r), \gamma_{A,2}^r(a_1^r)^0, \cdots, \gamma_{A,2}^r(a_1^r)^t,$ \\
      \hspace{2em} $\gamma_{A,2}^r(a_2^r)^0, \cdots, \gamma_{A,2}^r(a_2^r)^t)$ \hspace{6.4em}
      }};
      \node[anchor=center] at (mess to) {\shortstack[l]{\vspace{3.5em} \\
      $\gamma_{B,1}^{r^{\prime}}, \gamma_{B,2}^{r^{\prime}} \xleftarrow{R} \mathbb{Z}_q$\\
      $C_{B}^{r^{\prime}} = g_2^{\mathbf{w}_{B} \mathbf{B^{\star}}}$\\
      $\mathbf{w}_{B}^{r^{\prime}} = (\boldsymbol{\omega}_B^{r^{\prime}}, \gamma_{B,1}^{r^{\prime}}, 0)$\\
      $\boldsymbol{\omega}_{B}^{r^{\prime}} = (H(b_0^{r^{\prime}}), \gamma_{B,2}^{r^{\prime}}(b_1^{r^{\prime}})^0, \cdots, \gamma_{B,2}^{r^{\prime}}(b_1^{r^{\prime}})^t,$ \\
      \hspace{2em} $\gamma_{B,2}^{r^{\prime}}(b_2^{r^{\prime}})^0, \cdots, \gamma_{B,2}^{r^{\prime}}(b_2^{r^{\prime}})^t)$ \hspace{6.7em}
      }};        
      \end{sdblock}
      \postlevel
      \postlevel
      \end{sdblock}
      \vspace{3em}
      \begin{sdblock}{Query Phase}{}
      \begin{sdblock}{Join Query (Client)}{}
      \mess{A}{\shortstack[c]{ }}{B}
      \node[anchor=center] at (mess from) {\shortstack[l]{\vspace{5em} \\
      $k \xleftarrow{R} \mathbb{Z}_q\setminus{}\{0\}$\\
      $\delta_{A} \xleftarrow{R} \mathbb{Z}_q$\\
      $Tk_{A} = g_1^{\mathbf{v}_{A} \mathbf{B}}$\\
      $\mathbf{v}_{A} = (\boldsymbol{\nu}_{A}, 0, \delta_{A})$\\
      $\boldsymbol{\nu}_{A} = (k, p_{10}, \cdots, p_{1t}, \Vec{0})$, \\
      \hspace{1em} where $P_1(\phi_{1,j}) = 0$ for all $\phi_{1,j} \in \Phi_1$, $j \in [t]$ 
      }};
      \node[anchor=center] at (mess to) {\shortstack[l]{\vspace{5em} \\
      \vspace{1.5em}\\
      $\delta_{B} \xleftarrow{R} \mathbb{Z}_q$\\
      $Tk_{B} = g_1^{\mathbf{v}_{B} \mathbf{B}}$\\
      $\mathbf{v}_{B} = (\boldsymbol{\nu}_{B}, 0, \delta_{B})$\\
      $\boldsymbol{\nu}_{B} = (k, q_{10}, \cdots, q_{1t}, \Vec{0})$, \\
      \hspace{1em} where $Q_1(\psi_{1,j^{\prime}}) = 0$ for all $\psi_{1,j^{\prime}} \in \Psi_1$, $j^{\prime} \in [t]$
      }};
      \end{sdblock}
      \postlevel
      \postlevel      
      \postlevel
      \begin{sdblock}{Query Processing (Server)}{}
      \mess{A}{\shortstack[c]{ }}{B}
      \node[anchor=center] at (mess from) {\shortstack[l]{
       $D_{A}^r = e(Tk_{A}, C_{A}^r)$\\
       $D_{A}^r = e(g_1, g_2)^{det(B)kH(a_0^r)+P_1(a_1^r)}$ \hspace{4.9em}
       }};
      \node[anchor=center] at (mess to) {\shortstack[l]{
      $D_{B}^{r^{\prime}} = e(Tk_{B}, C_{B}^{r^{\prime}})$\\
      $D_{B}^{r^{\prime}} = e(g_1, g_2)^{det(B)kH(b_0^{r^{\prime}})+Q_1(b_1^{r^{\prime}})}$ \hspace{5.1em}
      }};        
      \end{sdblock}
      \begin{sdblock}{Query Result (Server)}{}
      \mess{A}{\shortstack[l]{ }}{B}
      \node[anchor=east] at (mess from) {\shortstack[l]{
      }};
      \node[anchor=east] at (mess to) {\shortstack[c]{\vspace{3em}\\
      If and only if $D_A^r == D_B^{r^{\prime}}$, then ``join'' takes place, since:\\
      the selection criterion are satisfied  \\
      ($\exists j, j^{\prime} \in [t]$ s.t. $a_1^r=\phi_{1,j}$ and $b_1^{r^{\prime}}=\psi_{1,j^{\prime}}$), \\
      and the join values of the two tables match ($a_0^r=b_0^{r^{\prime}}$). 
      }};        
      \end{sdblock}
      \postlevel
      \end{sdblock}
   \end{tikzpicture}
   \caption{Secure Join on $\mathbf{T_A, T_B}$ and the join query of Example \ref{ex:Sample}}
   \label{fig:SecJoin}
\end{figure*}
Our Secure Join scheme consists of five algorithms, namely (SJ.Setup, SJ.TokenGen, SJ.Enc, SJ.Dec, SJ.Match). The algorithms SJ.Setup, SJ.TokenGen, and SJ.Enc are applied by the client, on $T_{\tau}$, where $\tau \in \{A,B\}$. The server applies SJ.Dec to $T_{\tau}$. After applying the first four algorithms to both $T_A$ or $T_B$ and obtaining $D_A$ and $D_B$, the server applies the fifth algorithm, SJ.Match, to the results to find out whether $D_A$ and $D_B$ match. A positive answer allows performing the join operation.
\begin{enumerate}
    \item SJ.Setup$(1^\lambda)$: (Client, upload phase)\\
    On input the security parameter $\lambda$, the setup algorithm samples an asymmetric bilinear group $(\mathbb{G}_1,\mathbb{G}_2,\mathbb{G}_T,q,e)$ and chooses generators $g_1 \in \mathbb{G}_1$ and $g_2 \in \mathbb{G}_2$. Then, it samples $\mathbf{B} \gets \mathbb{GL}_n(\mathbb{Z}_q)$ and sets $\mathbf{B^*} = det(\mathbf{B}) \cdot (\mathbf{B}^{-1})^T$. Finally, the setup algorithm outputs the public parameters $pp= (\mathbb{G}_1,\mathbb{G}_2,\mathbb{G}_T,q,e)$ and the master secret key $msk= (pp,g_1,g_2,\mathbf{B},\mathbf{B^*})$.
    \item SJ.Enc($msk$, $\mathbf{w}_{\tau}^r$): (Client, upload phase) \\
    The encryption algorithm takes as input  the  master  secret key $msk$ and a vector $\mathbf{w}_{\tau}^r \in \mathbb{Z}^{m(t+1)+3}_q$ constructed from row $r$ in table $T_{\tau}$. To construct $\mathbf{w}_{\tau}^r$, the encryption algorithm chooses two uniformly random elements $\gamma^r_{\tau,1},\gamma^r_{\tau,2} \xleftarrow{R} \mathbb{Z}_q$ to form  $\mathbf{w}_{\tau}^r = ( \boldsymbol{\omega}^r_{\tau}, \gamma^r_{\tau,1}, 0)$. The vector $\boldsymbol{\omega}^r_{\tau}$, represents the information in the row $r$ of table $T_{\tau}$. This information, is the hash of the join value and $t$ powers of each of the other attribute values. Recall that $t$ is the same as the degree of polynomials introduced in Section \ref{sec:polynomials}. These powers of attributes values in $\boldsymbol{\omega}^r_{\tau}$ are obfuscated by $\gamma^r_{\tau,2}$ in $\mathbf{w}_{\tau}^r$. Therefore, for a sample row $r \in [n]$ in $T_A$ shown in Table \ref{Tab:Table_A}, we have $\boldsymbol{\omega}_{A}^r = (H(a_0^r), \gamma^r_{A,2} \cdot (a_1^r)^0, \cdots, \gamma^r_{A,2} \cdot (a_1^r)^t, \cdots, \gamma^r_{A,2} \cdot (a_m^r)^0, \cdots, \gamma^r_{A,2} \cdot (a_m^r)^t)$. In a similar way, for a row $r^{\prime} \in [n]$ of $T_B$ in Table \ref{Tab:Table_B}, we have $\boldsymbol{\omega}_{B}^{r^{\prime}} = (H(b_0^{r^{\prime}}), \gamma^{r^{\prime}}_{B,2} \cdot (b_1^{r^{\prime}})^0, \cdots, \gamma^{r^{\prime}}_{B,2} \cdot (b_1^{r^{\prime}})^t, \cdots, \gamma^{r^{\prime}}_{B,2} \cdot (b_m^{r^{\prime}})^0, \cdots, \gamma^{r^{\prime}}_{B,2} \cdot (b_m^{r^{\prime}})^t)$. The cryptographic hash function $H(\cdot)$ used in forming $\boldsymbol{\omega}_{A}^{r}$ and $\boldsymbol{\omega}_{B}^{r^{\prime}}$, maps each attribute values of the join column to a fixed-size value, and acts [as much as practically possible] like a random function.
    
     Now, SJ.Enc($\cdot$) is ready to perform the encryption, and computes $C^r_{\tau} = g_2^{\mathbf{w}^r_{\tau} \cdot \mathbf{B^{\star}}}$.
    \item SJ.TokenGen($msk$, $\Xi_{\tau}$): (Client, query phase)\\
    The token generation algorithm takes as input the master secret key $msk$
    and the join-query's filtering predicates for table $T_{\tau}$ shown by $\Xi_{\tau} = (\xi_{\tau,1}, \cdots, \xi_{\tau,m})$. Recall from Section \ref{sec:polynomials} that $\Xi_{A} = (\Phi_1, \cdots, \Phi_m)$ and $\Xi_{B} = (\Psi_1, \cdots, \Psi_m)$. We elaborated in Section \ref{sec:polynomials} that how the client chooses polynomials $P_i$'s and $Q_i$'s, $i \in [m]$, to encode the values specified in the \emph{IN} clauses $\Phi_i$'s and $\Psi_i$'s respectively. We also explained that for each $P_i$ or $Q_i$, there are at least $q$ such polynomials that the client can choose their candidate from. To generates a token for the join query, the SJ.TokenGen($\cdot$) algorithm chooses a uniformly random element $\delta_{\tau} \xleftarrow{R} \mathbb{Z}_q$ and generates a vector $\mathbf{v}_{\tau} \in \mathbb{Z}^{m(t+1)+3}_q$ of the form $\mathbf{v}_{\tau} = (\boldsymbol{\nu}_{\tau}, 0, \delta_{\tau})$. The vector $\boldsymbol{\nu}_{\tau}$ consists of a [non-zero] symmetric secret query key $k$ chosen randomly from $\mathbb{Z}_q \setminus{}\{0\}$ for encrypting the join attribute, and the coefficients of the polynomials corresponding to the filtering predicates $\Xi_{B} = (\Psi_1, \cdots, \Psi_m)$. Hence, to run the sample join query in Example \ref{ex:Sample}, the clients first needs to generate vectors
    $\boldsymbol{\nu}_A = (k, p_{1,0}, \cdots, p_{1,t}, \cdots, p_{m,0}, \cdots, p_{m,t})$ for table $T_A$, and $\boldsymbol{\nu}_B = (k, q_{1,0}, \cdots, q_{1,t}, \cdots, q_{m,0}, \cdots, q_{m,t})$ for table $T_B$. 
    Now, the token generation algorithm can compute $Tk_{\tau} = g_1^{\boldsymbol{\nu}_{\tau} \cdot \mathbf{B}}$, as the final token. 
    \item SJ.Dec(pp, $Tk_{\tau}$, $Ct_{\tau}^r$): (Server, query phase)\\
    On input of the public parameters $pp$, a token $Tk_{\tau}$, and a ciphertext $Ct_{\tau}^r$, the decryption algorithm computes $D^r_{\tau} = e(Tk_{\tau}, C^r_{\tau})$ for a row $r$ in $T_{\tau}$. The output of SJ.Dec($\cdot$), if the selection criteria is satisfied, equals $e(g_1, g_2)^{det(B)kH(a_0^r)}$ for table $T_A$. Similarly, for the same query, the decrypted value for row $r^{\prime}$ of table $T_B$ equals $e(g_1, g_2)^{det(B)kH(b_0^{r^{\prime}})}$, when the selection criteria is satisfied.
    
    There exist many (searchable) encryption schemes \cite{Curtmola} which can be used for pre-filtering the rows with the attributes matching the selection criteria reducing the size of the tables, but they are orthogonal to our join encryption scheme.  For a better exposition, we describe only the application of our encryption scheme.
    
    \item SJ.Match($D_A^r$, $D_B^{r^{\prime}}$): (Server, query result) \\
    This algorithm inspects the results of applying the previous four algorithms to all the rows in tables $T_A$ and $T_B$, and performs a join when there is a match. In other words, the decrypted value for each row $r$ in $T_A$, $D^r_A$, is compared with that of row $r^{\prime}$ in $T_B$, and if they match, the corresponding rows $r$ and $r^{\prime}$ in $T_A$ and $T_B$ are combined to form the row $(\theta^{r,r^{\prime}}, a_1^r, \cdots, a_m^r, b_1^{r^{\prime}}, \cdots, b_m^{r^{\prime}})$ in the join table, where $\theta^{r,r^{\prime}} = a_0^r = b_0^{\prime}$.
\end{enumerate}
\begin{example}\label{ex:SecJoin} 
Figure \ref{fig:SecJoin} shows the steps of Secure Join for the particular example of $T_A$ and $T_B$ in Tables \ref{Tab:Table_A} and \ref{Tab:Table_B}, and the join query in Example \ref{ex:Sample}. In the upload phase, the client first chooses the protocol parameters, and then starts the encryption for the (only) row $r$ in $T_A$ and the (only) row $r^{\prime}$ in $T_B$. The client forms the information vectors of these rows, $\boldsymbol{\omega}_A^r$ and $\boldsymbol{\omega}_B^{r^{\prime}}$, and prepares them for encryption by randomizing them to obtain vectors $\mathbf{w}_A^r$ and $\mathbf{w}_B^{r^{\prime}}$ that yield the final ciphertexts $C^r_A$ and $C^{r^{\prime}}_B$. The client uploads these ciphertexts at the server.

The client initiates the query phase by generating tokens to target rows with certain attribute values, specified in $\phi_1$ for $T_A$ and in $\psi_1$ for $T_B$, for the join operation. To do so, the client uses polynomial encoding to obtain vectors $\boldsymbol{\nu}_A$ and $\boldsymbol{\nu}_B$ for $\phi_1$ and $\psi_1$, then randomizes these vectors to $\mathbf{v}_A$ and $\mathbf{v}_B$, which yield the final tokens $Tk_A$ and $Tk_B$. Receiving these tokens in the query phase, the server decrypts the stored $C^r_A$ and $C^{r^{\prime}}_B$ with $Tk_A$ and $Tk_B$ to obtain $D^r_A$ and $D^{r^{\prime}}_B$. If $D^r_A \neq D^{r^{\prime}}_B$, the server disregards the query. However, if they do match, the server performs the join and combines the rows $r$ and $r^{\prime}$.
\end{example}

\section{Security}
\label{Sec:Security}

As stated in Section \ref{Sec:Intro}, we propose a new \emph{encryption} scheme for joins that not only restricts the leakage of the equality condition to tuples that match a \emph{selection criterion}, but also where the leakage of a series of queries corresponds to the transitive closure of the union of the leakage of each query, \emph{preventing super-additive leakage}.
Our design relies on the assumption that the clients are trusted and the server is semi-honest.
A semi-honest adversary wants to learn confidential data, but does not change queries issued by the application, query results, or the data in the DBMS.
This threat includes DBMS software compromises, root access to DBMS machines, and even access to the RAM of physical machines \citep{CryptDB}.

We structure our security proof as follows:
First, we prove that our modified inner-product encryption scheme from Section~\ref{Sec:Modifications} maintains the same security property as Kim et al.'s~\cite{InnerProductEnc}.
Then, using the simulator for the inner-product encryption we construct a simulator for our join encryption scheme.

\subsection{Inner Product Encryption}
\label{Sec:sec-ipe}

Kim et al.~prove security of their function-hiding inner-product encryption according to the following SIM-Security definition\footnote{An inner product encryption scheme that is secure under the simulation-based definition, is also secure under the indistinguishability-based definition. We refer to \citep{InnerProductEnc} for more details.}

\begin{definition}{(SIM-Security for Function-Hiding Inner-Product Encryption)}
\label{def:SIMSec_IPE}
Recall $\Pi_{ipe}$ = (IPE.Setup, IPE.KeyGen, IPE.Encrypt, IPE.Decrypt) from Section \ref{Sec:FHIPE}. $\Pi_{ipe}$ is SIM-secure if for all efficient adversaries $\mathcal{A}$, there exists and efficient simulator $\mathcal{S}$ that the output of the following experiments are computationally indistinguishable in security parameter $\lambda$.
\begin{enumerate}
    \item $Real_{\mathcal{A}}(1^{\lambda})$
    \begin{itemize}
        \item $(pp, msk) \leftarrow$ IPE.Setup($1^{\lambda}$)
        \item $b \leftarrow \mathcal{A}^{\mathcal{O}_{IPE.KeyGen}(msk, \cdot), \mathcal{O}_{IPE.Encrypt}(msk,\cdot)}(\lambda)$
        \item output b
    \end{itemize}
    \item $Sim_{\mathcal{A},\mathcal{S}}(1^{\lambda})$
    \begin{itemize}
        \item $(pp, st\footnote{\text{A simulator state}}) \leftarrow$ Setup$^{\prime}$($1^{\lambda}$)
        \item $b \leftarrow \mathcal{A}^{\mathcal{O}^{\prime}_{IPE.KeyGen}(st, \cdot), \mathcal{O}^{\prime}_{IPE.Encrypt}(st, \cdot)}(\lambda)$
        \item output b
    \end{itemize}
\end{enumerate}
The oracles $\mathcal{O}_{IPE.KeyGen}(msk, \cdot)$ and $\mathcal{O}_{IPE.Encrypt}(msk,\cdot)$ represent the real key generation and encryption oracles of $\Pi_{ipe}$, while $\mathcal{O}^{\prime}_{IPE.KeyGen}(st, \cdot)$ and $\mathcal{O}^{\prime}_{IPE.Encrypt}(st,\cdot)$ represent the simulated stateful key generation and encryption oracles.
\end{definition}
\begin{lemma}
\label{lem:sim-sec}
The modified function-hiding inner-product encryption scheme of Section~\ref{Sec:Modifications} is SIM-secure.
\end{lemma}

Kim et al.~\citep{InnerProductEnc} provide full details of achieving security through a simulation-based proof in a generic model of bilinear groups.
For space restrictions, we skip repeating the full detailed proof here, and provide their high-level idea instead.
We then discuss that the modifications we introduced (Section \ref{Sec:Modifications}) to their scheme does not compromise its security, hence the same security argument holds for our scheme as well.

To prove SIM-security of $\Pi_{ipe}$, Kim et al.~construct a generic bilinear group simulator $\mathcal{S}$ that interacts with the adversary $\mathcal{A}$, such that the distribution of responses in the real scheme is computationally indistinguishable from that in the ideal scheme. 
The simulator $\mathcal{S}$ must respond to the key generation and encryption queries as well as the generic bilinear group operation queries.
For each key generation and encryption query, the simulator responds with a fresh handle corresponding to each group  element in the secret key and the ciphertext. 
Similarly, for each generic group oracle query, the simulator responds with a fresh handle for the resulting group element.
The simulator maintains a table that maps handles to the formal polynomials the adversary forms via its queries.
Hence, each oracle query is regarded as referring to a formal query polynomial.
Two sets of formal variables are defined for $\Pi_{ipe}$, namely sets $\mathcal{R}$ and $\mathcal{T}$, with the universe $\mathcal{U}$ being the union of the two.
All the formal polynomials the adversary submits to the final test (zero-test) oracle are expressible in the formal variables in $\mathcal{R}$.
To answer the zero-test queries, the simulator performs a series of substitutions to re-express the adversary's query polynomials as a polynomial over the formal variables in $\mathcal{T}$.
The major challenge in the simulation is in answering the zero-test queries.
To consistently answer each zero-test query, the simulator first looks up the corresponding formal polynomial in its table and decomposes it into a “canonical” form, that is, as a sum of “honest”and “dishonest” components.
The honest components correspond to a proper evaluation of the inner product while the dishonest components include any remaining terms after the valid inner product relations have been factored out.
Kim et al.~argue, using properties of determinants, that if a query polynomial contains a dishonest component, then the resulting polynomial cannot be the identically zero polynomial over the formal variables corresponding to the randomly sampled elements in \textbf{B}.
Then, the simulator can correctly (with overwhelming probability) output “nonzero” in these cases.
Finally, in the ideal experiment, the simulator is given the value of the inner product between each pair of vectors the adversary submits to the key generation and encryption oracles, so it can make the corresponding substitutions for the honest inner product relations and thus, correctly simulate the outputs of the zero-test oracle. 

\begin{proof}
Being built on the scheme $\Pi_{ipe}$, the security of our scheme results from that of $\Pi_{ipe}$.
To prove the SIM-security of our Scheme in Section~\ref{Sec:Modifications}, we need to show that the modifications we introduced to the SIM-secure scheme $\Pi_{ipe}$, do not affect its security proof. 
\begin{enumerate}
    \item Changing the randomness from $\alpha$ and $\beta$, to $\delta$ and $\gamma_1$ in the input vectors \textbf{v} and \textbf{w} respectively.\\
    The simulator $\mathcal{S}$ should satisfy the following two conditions for corrections: i) $\mathcal{S}$'s response to the key generation, encryption, and group oracle queries made by the adversary $\mathcal{A}$ should be distributed identically as in the real experiment, and ii) $\mathcal{S}$ should correctly simulate the response to the zero-test queries made by the adversary $\mathcal{A}$. 
    While the latter is guaranteed by the properties of determinants and randomly sampled elements in the matrix \textbf{B}, the former is assured by the randomness of $\alpha$ and $\beta$.
    Our substitutes for the randomness provided by $\alpha$ and $\beta$, namely $\delta$ and $\gamma_1$, allow generating fresh handles in the simulation and consequently uniform and identical distributions as in the real experiment as per the first condition.
    The second condition however, is not affected by this modification, since it is satisfied by the properties of the matrix \textbf{B}, regardless of the values of $\alpha$ and $\beta$.
    Hence, replacing $\alpha$ and $\beta$ with ``1'', and including the randomness in the input vector does not affect the satisfaction of this condition. 
    \item Eliminating the first item of the pair for each of the followings: secret key, ciphertext, and decrypted value.\\
    We mentioned earlier in this section that all the formal polynomials the adversary submits to the final test (zero-test) oracle are expressible over the variables in the set $\mathcal{R}$.
    Kim et al.~argue that for any polynomial formed over these variables by the adversary, the simulator can correctly simulate the responses to the zero-test queries.
    Therefore, as long as new variables are not introduced to $\mathcal{R}$, this proof holds.
    Considering the elements in $\mathcal{R}$ (Definition 3.2 in \citep{InnerProductEnc}), eliminating the first item of the pair in the secret key, ciphertext and the decrypted value, does not change these formal variables, not to mention that it cannot introduce new variables to $\mathcal{R}$.
    Hence, it cannot enable the adversary to submit a ``dishonest'' component in the polynomial that outputs to zero in the zero-test to compromise the simulator.
\end{enumerate}
\end{proof}

Since our inner-product encryption scheme is SIM-secure we can replace its decryption keys and ciphertexts by outputs from the simulator and the two traces will be computationally indistinguishable as long as the decrypted plaintexts match.

\subsection{Secure Join Encryption}

We define security of our join encryption scheme following the methodology for symmetric searchable encryption (SSE) schemes by Curtmola et al.~\cite{Curtmola}.
Let $\lambda$ be the security parameter of the encryption schemes.
Let $H = \{ q_1, \ldots, q_\mu \}$ be a sequence of join queries where $\mu = \mathsf{poly}(\lambda)$.
Let $\sigma(q_i)$ be the result of the equi-join query $q_i$, i.e., $\sigma(q_i) = \{ (r_1, r'_1), \ldots, (r_\nu, r'_\nu)\}$ is the set of equality pairs between rows $r_{i,j}$ and $r'_{i,j}$ (which can be from the same table).
We define the trace $\tau(H) = \{ n, m, \sigma(q_1), \ldots, \sigma(q_\nu) \}$.

\begin{definition}{(SIM-Security for Join Encryption)}
\label{def:SIMSec_SJ}
We say a Join Encryption is SIM-secure, if there exists a simulator $\mathcal{S}(\tau(H))$ that given the trace $\tau(H)$ such that the following two experiments are computationally indistinguishable in the security parameter $\lambda$:
\begin{enumerate}
    \item $Real_{\mathcal{A}}(1^{\lambda})$
    \begin{itemize}
        \item $(pp, msk) \leftarrow$ IPE.Setup($1^{\lambda}$)
        \item $b \leftarrow \mathcal{A}(VIEW_{DBMS}(H))$
        \item output b 
    \end{itemize}
    \item $Sim_{\mathcal{A},\mathcal{S}}(1^{\lambda})$
    \begin{itemize}
        \item $(pp, st) \leftarrow$ Setup$^{\prime}$($1^{\lambda}$)
        \item $b \leftarrow \mathcal{A}(\mathcal{S}(\tau(H), st))$
        \item output b
    \end{itemize}
\end{enumerate}
\end{definition}

\begin{theorem}
\label{thm:sj-sec}
The scheme Secure Join = (SJ.Setup, SJ.TokenGen, SJ.Enc, SJ.Dec, SJ.Match) from Section~\ref{Sec:Scheme} is SIM-secure.
\end{theorem}

\begin{proof}
We construct the simulator $\mathcal{S}$ as follows:
Recall that $a^i_j$ and $b^i_j$ are the plaintext values in the encrypted tables.
We initialize all plaintext values to $\bot$.
Given a set of equality pairs $\sigma(q_h)$, we set all join values $a^i_0$ and $b^i_0$ of equality pairs to the same random values from $\mathbb{Z}_q$ preserving already set values not equal $\bot$.
All remaining join values of $\bot$ are replaced by random numbers.
We compute values for the selection attribute values $a^i_j$, $b^i_j$ and selection values $\phi_{j,i}$, $\psi_{j,i}$ by solving a linear system of equations for binary numbers 
-- one for each possible pair of attribute value (or selection value) and domain value from $[n]$.
This can be done in polynomial time using Gaussian elimination.
Finally, we encrypt all plaintexts for tables and generate keys for the queries using the simulator for inner-product encryption.
This results in random plaintexts with ciphertexts that produce the trace $\tau(H)$ as prescribed by the simulator.
Each query produces now decryption results $D = D^{\prime}$, if the join conditions (same query, same join values and selection criteria are satisfied) are fulfilled or random numbers otherwise.
The actual values of the query results cannot be obtained by the adversary, since they cannot break the discrete logarithm.

It remains to show that the real protocol also either produces the same decryption results $D$, $D'$ for join matches or random numbers, with overwhelming probability in $\lambda$ (note that $q = O(2^\lambda)$).
Without loss of generality, we simplify the representation and show by $D = e(g_1, g_2)^{det(B)kH(a_0)+P(a_1)}$ the decryption result for an arbitrary row from table $T_A$, with join value $a_0$ and selecting attribute $A_1$. $D$ is decrypted by a token generated by query key $k$ and polynomial $P(x)$. We show that if this $D$ equals a $D^{\prime}$ of the form $e(g_1, g_2)^{det(B)k^{\prime}H(b_0)+Q(b_1)}$ for an arbitrary row from $T_B$, decrypted 
 by a token generated by query key $k^{\prime}$ and polynomial $Q(y)$, then $D$ and $D^{\prime}$: i) belong to the same query, ii) have the same join value, and iii) satisfy the selection criteria. There are eight different cases to investigate based on the following conditions:

\begin{itemize}
    \item Satisfying/not-satisfying the selection criterion
    \item Belonging to the same/different query/ies
    \item Equality/non-equality of join values
\end{itemize}
\begin{claim}
    The equality $D=D^\prime$, holds with overwhelming probability if and only if all of the three conditions above are satisfied.
\end{claim}

In what follows, we show that this claim holds due to the randomness in: i) the symmetric key $k$, ii) the output of the hash function $H(\cdot)$, and iii) the coefficients of the polynomials, the probability of $D$ and $D^{\prime}$ taking the same value is negligible in the all cases except the case of belonging to the same query, having the same join value, and satisfying the selection criterion. In investigating the following cases we assume we are given a pair of decrypted values $(D, D^{\prime})$, where: 
    \begin{equation}
        \begin{aligned}
       & D = e(g_1, g_2)^{det(B)kH(a_0)+P(a_1)}, \\
       & D^{\prime} = e(g_1, g_2)^{det(B)k^{\prime}H(b_0)+Q(b_1)}.
        \end{aligned}
    \end{equation}

\begin{enumerate}
    \item Same query, same join values, both selection criterion hod,
    \begin{equation}
        Pr[D = D^{\prime}] = Pr[kH(a_0) = k^{\prime} H(b_0)] = 1.
    \end{equation}
    \item Same query, same join values, at least one of the selection criterion does not  hold, e.g.~the value of either of $a_1$ or $b_1$ is not included in the \emph{WHERE} clause,
    \begin{equation}
        \begin{aligned}
            \quad Pr[D = D^{\prime}] &= Pr[kH(a_0) + P(a_1) =k^{\prime} H(b_0) + Q(b_1)]
            \\&= Pr[P(a_1) = Q(b_1)] \leq \frac{t}{q}.
        \end{aligned}
    \end{equation}
    Recall from Section \ref{sec:polynomials} that how Lemma \ref{lem:Schwartz-Zippel} is applied to the polynomial encoding. The last inequality in the above equation results from applying Lemma \ref{lem:Schwartz-Zippel} to $P(x)-Q(b_1)$, which is a polynomial of degree $t$, and it evaluating to zero means the total probability of $P(a_1) = Q(b_1)$. This is an upper bound for the particular use case of $P(a_1) = Q(b_1)$. 
    \item Same query, different join values, both selection criterion hold, 
    \begin{equation}
        \begin{aligned}
            Pr[D = D^{\prime}] &= Pr[kH(a_0) = k^{\prime} H(b_0)] 
            \\&= Pr[H(a_0) = H(b_0)] = \frac{1}{q}.
        \end{aligned}
    \end{equation}
    $D$ and $D^{\prime}$ belonging to the same query results in $k=k^{\prime}$. Hence, we just need to calculate the probability of the hash output for join value $b_0$ colliding with that of $a_0$, while $a_0 \neq b_0$. As the cryptographic hash function $H(\cdot)$ acts as a random function with range $\mathbb{Z}_q$, this probability is $\frac{1}{q}$.
    \item Same query, different join values, at least one of the selection criterion does not hold,
    \begin{equation}
        \begin{aligned}
            Pr[D = D^{\prime}] &= Pr[kH(a_0) + P(a_1) =  k^{\prime}H(b_0) + Q(b_1)]
            \\&\leq \frac{t}{q}.
        \end{aligned}
    \end{equation}
    Similar to the item (2), the inequality above results from applying the Lemma \ref{lem:Schwartz-Zippel} to the polynomial $P(x) +kH(a_0) -  k^{\prime}H(b_0) - Q(b_1)$.
    \item Different queries, same join values, both selection criterion hold,
    \begin{equation}
        \begin{aligned}    
            Pr[D = D^{\prime}] &= Pr[kH(a_0) = k^{\prime}H(b_0)] \\
            &= Pr[k = k^{\prime}] + Pr[k  \neq k^{\prime}, H(a_0) = 0] \\
            &= \frac{1}{q-1} + \frac{q-2}{q-1} \times \frac{1}{q} =  \frac{2}{q}.
        \end{aligned}
    \end{equation}
    $D$ and $D^{\prime}$ correspond to equal join values, hence $H(a_0) = H(b_0)$. If $H(a_0)$ is not zero (which has probability $\frac{q-1}{q}$), the equality of $D$ and $D^{\prime}$ requires the equality of the query keys, i.e., $k = k^{\prime}$. Since $k$ and $k^{\prime}$ belong to different queries are chosen independently and randomly from $\mathbb{Z}_q\setminus\{0\}$, the probability of one of them taking the same value as the other one is $\frac{1}{q-1}$, resulting in the overall probability $\frac{2}{q}$ for $D=D^{\prime}$.
    \item Different queries, same join values, at least one of the selection criterion does not hold,
    \begin{equation}
        \begin{aligned} 
            Pr[D = D^{\prime}] &= Pr[kH(a_0) + P(a_1) = k^{\prime}H(b_0) + Q(b_1)]
            \\&\leq \frac{t}{q}.
        \end{aligned}
    \end{equation}    
    Similar to the items (2) and (4), the inequality above results from applying the Lemma \ref{lem:Schwartz-Zippel} to the polynomial $P(x) +kH(a_0) - k'H(b_0) - Q(b_1)$.    
    \item Different queries, different join values, both selection criterion hold, 
    \begin{equation}
        \begin{aligned} 
            Pr[D = D^{\prime}] = Pr[kH(a_0) = k^{\prime}H(b_0)] = \frac{1}{q}.
        \end{aligned}
    \end{equation}             
    We calculated the probability of $H(a_0)=H(b_0)$ when $a_0 \neq b_0$ in item 3, which is the probability of a $H(b_0)$ taking the same random value as the other independent value $H(a_0)$. This probability does not increase by multiplying these random values by other fixed or random values ($k$ and $k^{\prime}$ here). 
    \item Different queries, different join values, at least one of the selection criterion does not hold,
    \begin{equation}
        \begin{aligned} 
            Pr[D = D^{\prime}] &= Pr[kH(a_0) + P(a_1) = k^{\prime}H(b_0) + Q(b_1)] 
            \\&\leq \frac{t}{q}.
        \end{aligned}
    \end{equation} 
    Similar to the items (2), (4), and (6), the inequality above results from applying the Lemma \ref{lem:Schwartz-Zippel} to the polynomial $P(x) + kH(a_0) - k'H(b_0) - Q(b_1)$. 
\end{enumerate}
\textbf{From a given $\mathbf{(D, D^{\prime})}$ to any $\mathbf{(D, D^{\prime})}$ in the set of queries.}\\
Through items (1) - (8), we showed that for a given $(D, D^{\prime})$, the equality holds if and only if the corresponding rows to these values are decrypted through processing the same query, have the same join value, and satisfy the selection criteria in the query. The probability of the equality $D = D^{\prime}$ taking place for any other cases is in $O(\frac{1}{q})$. However, if we consider all the 
$\varepsilon \leq 2 \mu n$ 
ciphertext decryptions 
by $\mu$ queries over the join tables, this probability can increase to  
$O(\frac{\varepsilon^2}{2q})$,
according to the birthday paradox. However, $O(\sqrt{q} = 2^{\lambda/2})$ is a loose upper bound for $\varepsilon$ in our scheme, since the number of queries and the number of database rows in our scheme is polynomial-sized in $\lambda$, i.e., 
$\varepsilon = poly(\lambda)$.
\end{proof}

We can re-iterate our security properties as corollaries of Theorem~\ref{thm:sj-sec}.

\begin{corollary}{(Restricting leakage to the selection criterion)} \label{thm:SJSelect}
The Secure Join = (SJ.Setup, SJ.TokenGen, SJ.Enc, SJ.Dec, SJ.Match) restricts the leakage to the selection criterion.
\end{corollary}

The decryption algorithm, SJ.Dec($\cdot$), outputs a value of the form $D = e(g_1, g_2)^{det(B)kH(a_0)+P(a_1)}$. Hence, the server can obtain $D = e(g_1, g_2)^{det(B)kH(a_0)}$, if and only if $P(a_1)$ evaluate to zero. According to Lemma \ref{lem:Schwartz-Zippel} the probability of $P(a_1)$ evaluating to zero when the values of $a_1$ is not in the \emph{WHERE} clause, is negligible. Hence, Secure Join successfully restricts the leakage to the selection criterion.

\begin{corollary}{(Preventing Super-additive leakage)} \label{thm:SJLink}
The Secure Join = (SJ.Setup, SJ.TokenGen, SJ.Enc, SJ.Dec, SJ.Match) prevents super-additive leakage.
\end{corollary}

Items (5) to (8) in the proof of Theorem \ref{thm:sj-sec} show that the probability that SJ.Match results in a match for different queries is negligible.
This protection holds even if both selection polynomials evaluate to zero, or even if the join values match.
Hence, Secure Join successfully prevents the adversary from linking the results of different queries.

\oddsidemargin -0.25in
\textwidth 7.0in         
\topmargin 0.0in
\headheight 0.0in
\headsep 0.0in
\topskip 0.0in
\footskip 0.4in
\textheight 8.8in         
\section{Experiments}\label{Sec:Experiments}
In this section, we evaluate our Secure Join scheme (introduced in Section \ref{Sec:Scheme}) in running secure hash joins over outsourced data. We provide three evaluation categories: i) a benchmark of the cryptographic operations in Secure Join, ii) server's performance in performing joins operation (decryption anch match) for various database sizes, and iii) server's performance in performing joins operation for various \emph{IN} clause sizes for a single attribute. 
\subsection{Setup} 
We ran our experiments using a single thread on a machine with four processors, a 64-bit Intel Core $i7-7500U @ 2.70GHz$ each, with $15.4 GiB$ RAM and running Ubuntu $20.04.01$.
In our experiments, we used an adjusted version of function-hiding inner product encryption implementation by Kim et al.~ \cite{InnerProductEnc} (described in Section \ref{Sec:FHIPE}). The adjustments were made to the implementation, so that it complies with our modifications from Section~\ref{Sec:Modifications}. In our experiments, we use the data provided by TPC-H benchmark. We in particular use the tables tables $T_A$ = \emph{Orders} and $T_B$ = \emph{Customers} from the benchmark. The table \emph{Orders}, has nine attribute values: \textit{(orderkey, custkey, orderstatus, totalprice, orderdate, orderpriority, clerk, shippriority, comment)}, while \emph{Customers} has eight: \textit{(custkey, name, address, nationkey, phone, acctbal, mktsegment, comment)}. The original tables \emph{Orders} and \emph{Customers} have $150,000$ and $1,500,000$ rows respectively. We use scale factors in the range of (0.01 - 0.1).
The customer key information provides the join attribute \textit{custkey} in both tables. We also add another attribute column \textit{selectivity} to both tables. \textit{Selectivity} takes values $\{\frac{1}{12.5}, \frac{1}{25}, \frac{1}{50}, \frac{1}{100}\}$. We use these values for two purposes: i) providing some values for the attribute \textit{Selectivity}, ii) showing the proportion of the table assigned with this attribute value. Hence, each \textit{Selectivity} value $x$ is assigned to $x\times n$ rows where $n$ is the table size. 
For example, in a table \emph{Customers} with $150,000$ rows, $1,500$ rows are assigned the same attribute value with $\frac{1}{100}$ as their \textit{Selectivity}. 

\subsection{Crypto Operations in Secure Join}\label{Sec:Crypto}
Figure \ref{fig:Operations} shows a micro-benchmark of our implementation of the cryptographic operations in Secure Join for a single row. 
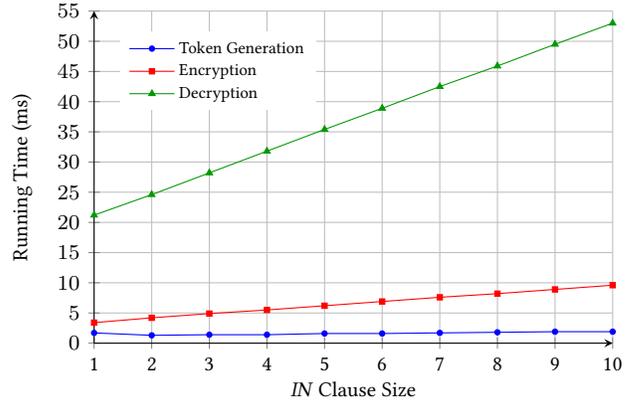
\begin{figure}[h]
    \centering
    \pgfplotsset{scaled x ticks=false}
    \pgfplotsset{scaled x ticks=false}
    \begin{tikzpicture}
    \begin{axis}[
        width=1\linewidth,
        height=6cm,
        axis lines=left,
        grid=both,
        legend style={font=\small},
        legend cell align=left,
        legend style={at={(0.24, 0.70)}, anchor=south, fill=white, draw=none},
        label style={font=\small},
        legend style={nodes={scale=0.8, transform shape}},
        ticklabel style={font=\small},
        xlabel style={at={(0.5, 0.03)}},
        ylabel style={at={(0.04, 0.5)}},
        ymin=0,
        ymax=55, 
        ytick={0,5,10,...,55},
        xmin=1,
        xmax=10,
        xtick={1, 2, 3, 4, 5, 6, 7, 8, 9, 10},
        xticklabels={1, 2, 3, 4, 5, 6, 7, 8, 9, 10},
        xlabel=\emph{IN} Clause Size,
        ylabel=Running Time (ms),
        legend entries={{Token Generation},{Encryption},{Decryption}},
    ]
    \addplot[color=blue, mark=*, mark options={scale=0.5}] coordinates {
    (1, 1.7) (2, 1.3) (3, 1.4) (4, 1.4) (5, 1.6) (6, 1.6) (7, 1.7) (8, 1.8) (9, 1.9) (10, 1.9)};
    \addplot[color=red, mark=square*, mark options={scale=0.5}] coordinates {
    (1, 3.4) (2, 4.2) (3, 4.9) (4, 5.5) (5, 6.2) (6, 6.9) (7, 7.6) (8, 8.2) (9, 8.9) (10, 9.6)};        
    \addplot[color=black!40!green, mark=triangle*, mark options={scale=0.7}] coordinates {
    (1, 21.2) (2, 24.6) (3, 28.2) (4, 31.8) (5, 35.4) (6, 38.9) (7, 42.5) (8, 45.9) (9, 49.5) (10, 53)};
    \end{axis}
    \end{tikzpicture}
    \caption{Encryption operation benchmarks for a single row in table \emph{Customers} 
    }
    \label{fig:Operations}
\end{figure}
This experiment provides the average implementation results for a row in the table \emph{Customers}, 
when the \emph{IN} clause size varies from 1 to 10. The cryptographic operations are SJ.Enc($\cdot$), SJ.TokenGen($\cdot$), and SJ.Dec($\cdot$) from Section \ref{Sec:Scheme}. The token generation algorithm does not show a noticeable change in runtime over different values for size of the \emph{IN} clause ($t$ in Section \ref{sec:polynomials}). The algorithm takes less than $2 ms$ to run for each value of $t$, since it only calculates a single value of $g_1^{\boldsymbol{\nu}_{\tau} \mathbf{B}}$, although increasing $t$ from 1 to 10 changes the non-zero values in the $\boldsymbol{\nu}$ vector from 2 to 11 respectively. The encryption algorithm takes $3.4 ms$ on average to encrypt a rows for $t=1$, this time increases linearly to $9.6 ms$ for $t=10$, since the algorithm calculates the $t$ powers for each attribute value in the table, pre-encryption. The most time consuming cryptographic operation in Secure Join is the decryption algorithm, which takes $21.2 ms$ to run for $t=1$ which increases to $53 ms$ for $t=10$.  
\subsection{Joins and Database Size}\label{Sec:ScaleFactor}
Figure \ref{fig:ScaleFac} shows the runtime of joins operation over the encrypted data on the server side, i.e.~SJ.Dec($\cdot$) and SJ.Match($\cdot$), for several database sizes and different \textit{Selectivity} values ($s$), when the join query consists of a single value in the $IN$ clause for each table. 
\begin{figure}[h]
    \centering
    \pgfplotsset{scaled x ticks=false}
    \pgfplotsset{scaled x ticks=false}
    \begin{tikzpicture}
    \begin{axis}[
        width=1\linewidth,
        height=6cm,
        axis lines=left,
        grid=both,
        legend style={font=\small},
        legend cell align=left,
        legend style={at={(0.17, 0.62)}, anchor=south, fill=white, draw=none},
        label style={font=\small},
        legend style={nodes={scale=0.8, transform shape}},
        ticklabel style={font=\small},
        xlabel style={at={(0.5, 0.03)}},
        ylabel style={at={(0.04, 0.5)}},
        ymin=0,
        ymax=300, 
        ytick={0,25, 50,...,300},
        xmin=0.01,
        xmax=0.1,
        xtick={0.01, 0.02, 0.03, 0.04, 0.05, 0.06, 0.07, 0.08, 0.09, 0.1},
        xticklabels={0.01, 0.02, 0.03, 0.04, 0.05, 0.06, 0.07, 0.08, 0.09, 0.1},
        xlabel=TPC-H Scale Factor,
        ylabel=Running Time (s),
        legend entries={{$s = 1/100$}, {$s=1/50$}, {$s=1/25$}, {$s=1/12.5$}},
    ]
    \addplot[color=blue, mark=*, mark options={scale=0.5}] coordinates {
  (0.01, 3.52) (0.02, 7.06) (0.03, 10.52) (0.04, 14.16) (0.05, 17.60) (0.06, 21.25) (0.07, 24.40) (0.08, 28.31) (0.09, 31.82) (0.1, 35.34)};
    \addplot[color=red, mark=square*, mark options={scale=0.5}] coordinates {
  (0.01, 7.01) (0.02, 14.03) (0.03, 20.99) (0.04, 28.26) (0.05, 35.04) (0.06, 42.54) (0.07, 48.75) (0.08, 56.52) (0.09, 63.60) (0.1, 70.62)};
    \addplot[color=black!40!green, mark=triangle*,  mark options={scale=0.7}] coordinates {
  (0.01, 13.99) (0.02, 27.96) (0.03, 42.12) (0.04, 56.57) (0.05, 69.92) (0.06, 84.78) (0.07, 97.53) (0.08, 113.10) (0.09, 127.05) (0.1, 141.31)};
    \addplot[color=black!40!purple, mark=star, mark options={scale=0.8}] coordinates {
  (0.01, 27.88) (0.02, 55.82) (0.03, 83.90) (0.04, 112.96) (0.05, 139.89) (0.06, 169.21) (0.07, 194.99) (0.08, 225.97) (0.09, 254.46) (0.1, 282.49)};
    \end{axis}
    \end{tikzpicture}
    \caption{Joins runtime for various scale factors, single \emph{IN} clause}
    \label{fig:ScaleFac}
\end{figure}
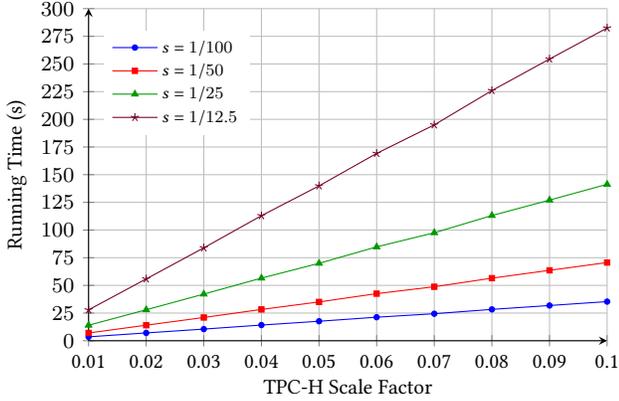
As expected, the runtime of performing the joins operation increases linearly with the database size. This increase however, is more noticeable for higher values of \textit{Selectivity}. When the \textit{Selectivity} value, shown by $s$ in Figure \ref{fig:ScaleFac} is $\frac{1}{100}$, the server takes $3.52 s$ to run joins over tables \emph{Orders} and \emph{Customers} with scale factor $0.01$ and $35.34 s$ to do the same for the tables with scale factor $0.1$. However, when $s=\frac{1}{12.5}$, the server takes $27.88 s$ to run joins over \emph{Orders} and \emph{Customers} with scale factor $0.01$ and $282.49 s$ to do so for the tables with scale factor $0.1$. 

\subsection{Joins and IN-Clause Size}\label{Sec:IN-Clause}
Figure \ref{fig:IN_Size} shows the runtime of joins operation over the encrypted data on the server side, i.e.~SJ.Dec($\cdot$) and SJ.Match($\cdot$), for several sizes of \emph{IN} clause ($t$) and different \textit{Selectivity} values ($s$), when the scale factor for join tables \emph{Orders} and \emph{Customers} is $0.01$. 
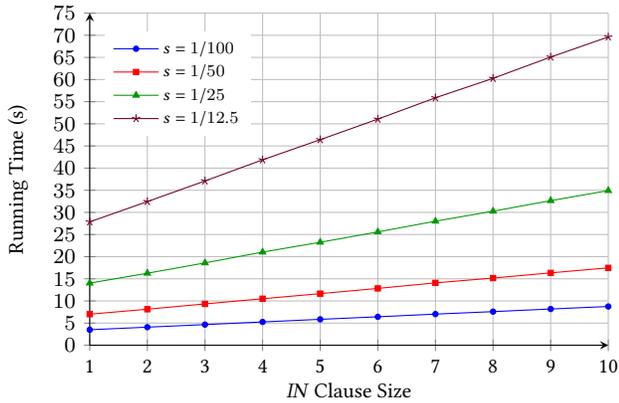
\begin{figure}[h]
    \centering
    \pgfplotsset{scaled x ticks=false}
    \begin{tikzpicture}
    \begin{axis}[
        width=1\linewidth,
        height=6cm,
        axis lines=left,
        grid=both,
        legend style={font=\small},
        legend cell align=left,
        legend style={at={(0.17, 0.63)}, anchor=south, fill=white, draw=none},
        label style={font=\small},
        legend style={nodes={scale=0.8, transform shape}},
        ticklabel style={font=\small},
        xlabel style={at={(0.5, 0.03)}},
        ylabel style={at={(0.04, 0.5)}},
        ymin=0,
        ymax=75, 
        ytick={0,5, 10,...,75},
        xmin=1,
        xmax=10,
        xtick={1, 2, 3, 4, 5, 6, 7, 8, 9, 10},
        xticklabels={1, 2, 3, 4, 5, 6, 7, 8, 9, 10},
        xlabel=\emph{IN} Clause Size,
        ylabel=Running Time (s),
        legend entries={{$s = 1/100$}, {$s=1/50$}, {$s=1/25$}, {$s=1/12.5$}},
    ]
    \addplot[color=blue, mark=*, mark options={scale=0.5}] coordinates {
  (1, 3.50) (2, 4.08) (3, 4.66) (4, 5.26) (5, 5.85) (6, 6.42) (7, 7.03) (8, 7.59) (9, 8.18) (10, 8.75)};
    \addplot[color=red, mark=square*, mark options={scale=0.5}] coordinates {
  (1, 7.02) (2, 8.13) (3, 9.33) (4, 10.50) (5, 11.65) (6, 12.83) (7, 14.08) (8, 15.17) (9, 16.34) (10, 17.47)};
    \addplot[color=black!40!green, mark=triangle*, mark options={scale=0.7}] coordinates {
  (1, 14.03) (2, 16.26) (3, 18.59) (4, 21.03) (5, 23.25) (6, 25.59) (7, 28.01) (8, 30.27) (9, 32.65) (10, 34.92)};
    \addplot[color=black!40!purple, mark=star, mark options={scale=0.8}] coordinates {
  (1, 27.86) (2, 32.44) (3, 37.10) (4, 41.87) (5, 46.42) (6, 51.07) (7, 55.88) (8, 60.27) (9, 65.08) (10, 69.62)};
    \end{axis}
    \end{tikzpicture}
    \caption{Joins runtime for \emph{IN} clause with various sizes, scale factor: 0.01}
    \label{fig:IN_Size}
\end{figure}
As depicted in the plots in Figure \ref{fig:IN_Size}, increasing $t$, results in a longer runtime for Joins. This increase, however taking place in all \textit{Selectivity} experiments, is more noticeable the value of $s$ is larger. It takes the server $3.50 s$ to run the joins operation over the tables \emph{Orders} and \emph{Customers} for $t=1$ and $8.75 s$ to do the same for $t=10$, when $s=\frac{1}{100}$. The corresponding running time for $s=\frac{1}{12.5}$ are $27.86 s$ and $69.62 s$.      

We ran each of the experiments in Sections \ref{Sec:ScaleFactor} and \ref{Sec:IN-Clause}, 25 times, aiming to demonstrate the data in the 95\% confidence interval. However, the deviations from the mean value in our results were of order $10^{(-2)}$, resulting in unobservable error bars in our plots in Figures \ref{fig:ScaleFac} and \ref{fig:IN_Size}. 

\subsection{Comparison and Discussion}

We mentioned in Section \ref{Sec:Intro} that a state-of-the-art encryption scheme for joins by Hahn et al.~\cite{Joins} reduces the leakage to only leaking the equality condition for tuples that match a selection criterion. However, their scheme: i) requires nested-loop joins (with time complexity $O(n^2)$), ii) only works for primary key, foreign key joins, and iii) still results in a super-additive leakage, as we showed in Section \ref{ex:problem}. Our new encryption scheme for joins prevents this super-additive leakage, is not limited to primary key, foreign key joins, and can run hash joins (with expected time complexity $O(n)$). 

It is challenging to provide a one-to-one comparison of our performance measurement experiments with those in \cite{Joins}, due to: i) the differences in the parameters (the number of rows and attribute values) of the join tables, ii) the unclarity of join query parameters such as \emph{IN} clause size and selectivity in \cite{Joins}, and iii) different hardware used to perform the experiments.
Hence, we provide approximate performance comparisons. Their experiments \cite{Joins} report an average runtime of $15$ seconds for $1000$ average decrypted values, i.e., $15 ms$ per decryption. Our results in Figure \ref{fig:Operations} show an average time of $21 ms$ for one decryption operation (for \emph{IN} clause size of one). The results in \cite{Joins} also report the average of $6 s$ for a join operation over tables \emph{Part} ($20,000$ rows and 6 attribute values) and \emph{LineItem} ($6,000,000$ rows and 8 attribute values) from TPC-H, with scale factor $0.1$. Our experiments show a result of $35 s$ over tables \emph{Orders} and \emph{Customers} with scale factor 0.1 (and selectivity $\frac{1}{100}$).  In conclusions, our performance is already on the same order of magnitude even without any parallelization at better security.

Furthermore, the experiments in \cite{Joins} benefit from performance improvement provided by parallelizing each join query over 32 cores. The authors also re-use the decrypted information of the earlier join queries in the later ones to boost performance further. While our experiments are not intended to reuse the decrypted information due to the stronger security objectives of our scheme, they can as well benefit from parallelizing over several cores, instead of running on just one (the current setup).
\section{Related Work}
\label{Sec:RelatedWork}

In this paper we consider non-interactive equi-joins over encrypted data in a single-client, single-server setting.
This model is known as the database-as-a-service model~\cite{hacigumus2002providing}.
We provide the history of join encryption schemes \cite{Joins,CryptDB,AnalQoED,SQLoED} that address security in this model in Section~\ref{Sec:Problem}.
All of these schemes use a deterministic or searchable encryption scheme as its basic building block and then aim to reduce the leakage of the join pattern, i.e., the number of equality pairs revealed.
We argue that our join encryption scheme proposed in this paper leaks the least information in this setting and represents a natural lower bound of the necessary leakage in the two table setting.

CryptDB~\cite{CryptDB} also introduced the concept of re-encryption as a method to reduce leakage in the multiple (more than two) table setting.
The idea of re-encryption is that each table is encrypted with a different key and tables are re-encrypted to joint keys on a join operation.
Kerschbaum et al.~\cite{kerschbaum2013optimal} introduce an algorithm that optimizes the selection of the joint key.
Mironov et al.~ \cite{TCC17} present a new encryption scheme that makes the re-encryption uni-directional and hence prevents linking non-matched rows in a group of joins.
Note that our proposed encryption scheme uses a fresh key in each query and hence does not need to resort to any of these techniques in order to achieve their (and stronger) security properties.

Pang and Ding \cite{pang2014privacy} present a secret-key encryption scheme that avoids self-joins and Wang and Pang \cite{bilinearjoins} extend it to a public-key encryption scheme.
Carbunar and Sion \cite{bloomfilterjoin} use Bloom filters to achieve a similar security guarantee, but also have to manage false positives at the client due to the properties of Bloom filters.
All three schemes either reveal the equality pairs of the entire columns or none, just as CryptDB.
Hahn et al.~\cite{Joins} and now our scheme improve over this by only revealing the equality pairs for rows matching a selection criterion.

Many join algorithms can be parallelized, and Bultel et al.~\cite{bultel2018secure} process joins over encrypted data in a map-reduce cluster.
Such a parallelization is also applicable to our encryption scheme and could further speed up join operations over large data, although our performance is already competitive to the state-of-the-art at better security.

Interactive schemes, e.g., secure multi-party computation, fully homomorphic encryption with intermediate decryption or secure hardware, can also implement joins over encrypted data.
These schemes need to implement a circuit, i.e., an algorithm whose instructions and data accesses are independent of the input data.
Agrawal et al.~\cite{sovereignjoins} were the first to introduce this problem as sovereign joins.
Arasu and Kaushik \cite{obliviousquery} present the first, non-trivial, secure algorithms, but they are still too complicated to be practically implemented.
Krastnikov et al.~\cite{simeon} present the first non-trivial, secure and practical algorithms.

Joins can also be consider between datasets from multiple parties.
Mykletun and Tsudik \cite{mykletun2006security} are the first to point out there is inherent threat of collusion between one of the (two) parties and the database provider which cannot be fully avoided.
Hang et al.~\cite{enki} present a key management scheme for multiple parties in this setting.
Kantarcioglu et al.~\cite{kantarcioglu2009formal} develop an anonymization scheme that can prevent some of the leakage while maintaining efficiency.

A cryptographic technique to match datasets from two parties is private set intersection (PSI).
PSI was introduced by Fagin et al.~\cite{fagin1996comparing} and formally developed by Freedman et al.~\cite{freedman2004efficient}.
It is practically deployed by Google and Mastercard \cite{ion2020deploying}.
PSI protocols using a service provider exist \cite{kerschbaum2012outsourced,abadi2015psi,kerschbaum2012collusion}.
However, PSI cannot be applied to equi-join, since a prerequisite for PSI is that each data element in each set is unique.
The security of most PSI protocols deteriorates to the all-or-nothing disclosure level of CryptDB when elements can be replicated as in an equi-join over two database tables.

\section{Conclusion}
\label{Sec:Conclusions}

In this paper we present a new join encryption scheme that prevents additional leakage from a series of queries.
We compare our scheme to the state-of-the-art join encryption schemes and it achieves comparable performance even without parallelization, better scalability (due to hash joins instead of nested loop joins) and reduced leakage over series of queries, i.e., better security.
We provide a formal security proof and evaluate an implementation over a dataset from the TPC-H benchmark.
We claim our construction achieves a natural lower bound for the leakage in an efficient, non-interactive, single-server setting.
Hence, future work can investigate which restrictions to remove in order to further reduce the leakage of join encryption schemes.
\bibliographystyle{plain}
\bibliography{references}
\end{document}